\documentclass{aptpub}

\usepackage{amsmath}
\usepackage{amssymb}
\usepackage{amscd}
\usepackage{bbm}
\usepackage{calc}
\usepackage{cite}
\usepackage{graphicx}
\usepackage[colorlinks=false]{hyperref}
\usepackage{natbib}
\usepackage{parskip}
\usepackage{placeins}
\usepackage{xcolor}
\usepackage[linesnumbered]{algorithm2e}
\usepackage{geometry}
\newcommand{\mvbar}{\middle\vert} 
\newcommand{\ud}{\,\mathrm{d}} 

\newcommand{\target}[3]{f^{#1}_{#2}\!\left(#3\right)} 
\newcommand{\mctarget}[3]{\tilde{f}^{#1}_{#2}\!\left(#3\right)} 
\newcommand{\termprop}[3]{h^{\text{#1}}_{#2}\!\left(#3\right)} 
\newcommand{\Trans}[4]{p^{\text{#1}}_{#2}\!\left(#4 \,\mvbar\, #3\right)} 

\newcommand{\lange}[2]{\mathbb{DL}^{#1}_{#2}} 
\newcommand{\ornst}[2]{\mathbb{O}^{#1}_{#2}} 

\newcommand{\A}[3]{A^{\text{#1}}_{#2}\left(#3\right)} 
\newcommand{\Bmu}{\hat{\boldsymbol{\mu}}_c} 
\newcommand{\cores}{^{C}_{c=1}} 
\newcommand{\phifn}[3]{\phi^{\text{#1}}_{#2}\!\left(#3\right)} 
\newcommand{\Phifn}[2]{\Phi^{\text{#1}}_{#2}} 
\newcommand{\tX}{\ensuremath{\mathfrak{X}}\xspace} 
\newcommand{\x}{\boldsymbol x} 
\newcommand{\X}[2]{\boldsymbol X_{#1}^{\ifthenelse{\equal{#2}{}}{}{(#2)}}} 
\newcommand{\y}{\boldsymbol y} 
\newcommand{\yT}{\boldsymbol y_T} 

\newcommand{\algref}[1]{\hyperref[#1]{Algorithm \ref*{#1}}}
\newcommand{\stepref}[1]{\hyperref[#1]{Step \ref*{#1}}}
\newcommand{\algstref}[2]{\hyperref[#2]{Algorithm \ref*{#1} Step \ref*{#2}}}
\newcommand{\figref}[1]{\hyperref[#1]{Figure \ref*{#1}}}
\newcommand{\subfigref}[3]{\hyperref[#1]{Figure \ref*{#2}(#3)}}
\newcommand{\tabref}[1]{\hyperref[#1]{Table \ref*{#1}}}
\newcommand{\apxref}[1]{\hyperref[#1]{Appendix \ref*{#1}}}
\newcommand{\apxrefpl}[2]{Appendices \hyperref[#1]{\ref*{#1}} and \hyperref[#2]{\ref*{#2}}}
\newcommand{\secref}[1]{\hyperref[#1]{Section \ref*{#1}}}
\newcommand{\prinref}[1]{\hyperref[#1]{Principle \ref*{#1}}}
\newcommand{\conref}[1]{\hyperref[#1]{Condition \ref*{#1}}}
\newcommand{\resref}[1]{\hyperref[#1]{Result \ref*{#1}}}
\newcommand{\defnref}[1]{\hyperref[#1]{Definition \ref*{#1}}}
\newcommand{\thmref}[1]{\hyperref[#1]{Theorem \ref*{#1}}}
\newcommand{\lemref}[1]{\hyperref[#1]{Lemma \ref*{#1}}}
\newcommand{\corrolref}[1]{\hyperref[#1]{Corollary \ref*{#1}}}
\newcommand{\remref}[1]{\hyperref[#1]{Remark \ref*{#1}}}

\newcommand{\W}[2]{\boldsymbol W_{#1}^{\ifthenelse{\equal{#2}{}}{}{(#2)}}} 
\newcommand{\mc}{\boldsymbol m_c} 
\newcommand{\Vc}{\mathbf V_c} 
\newcommand{\D}{\mathbf D} 
\newcommand{\bH}{\mathbf H} 
\newcommand{\bM}{\mathbf M} 
\def\hatprecon{\hat{\mathbf{\Lambda}}} 

\authornames{Dai, Pollock, Roberts} 
\shorttitle{Monte Carlo Fusion} 

\begin{document}

\title{Monte Carlo Fusion}

\authorone[University of Essex]{Hongsheng Dai} 
\authortwo[University of Warwick]{Murray Pollock} 
\authorthree[University of Warwick]{Gareth Roberts} 
\addressone{Department of Mathematical Sciences, University of Essex, Wivenhoe Park, Colchester, CO4 3SQ} 
\emailone{hdaia@essex.ac.uk}
\addresstwo{Department of Statistics, University of Warwick, Gibbet Hill Road, Coventry, CV4 7ES} 
\emailtwo{m.pollock@warwick.ac.uk; gareth.o.roberts@warwick.ac.uk}

$\vspace{-3cm}$
\begin{abstract}
This paper proposes a new theory and methodology to tackle the problem of unifying distributed analyses and inferences on shared parameters from multiple sources, into a single coherent inference. This surprisingly challenging problem arises in many settings (for instance, expert elicitation, multi-view learning, distributed `big data' problems etc.), but to-date the framework and methodology proposed in this paper (Monte Carlo Fusion) is the first general approach which avoids any form of approximation error in obtaining the unified inference. In this paper we focus on the key theoretical underpinnings of this new methodology, and simple (direct) Monte Carlo interpretations of the theory. There is considerable scope to tailor the theory introduced in this paper to particular application settings (such as the big data setting), construct efficient parallelised schemes, understand the approximation and computational efficiencies of other such unification paradigms, and explore new theoretical and methodological directions.
\end{abstract}





\section{Introduction} \label{sec:intro}

A common problem arising in statistical inference is the need to unify distributed analyses and inferences on shared parameters from multiple sources, into a single coherent inference. This unification (or what we term `fusion') problem can arise either explicitly due to the nature of a particular application, or artificially as a consequence of the approach a practitioner takes to tackling an application.

Typically there will exist no closed form analytical approach to unifying distributed inferences, and so we focus on a Monte Carlo approach. Stated generally, in this paper we are interested in sampling (without error) the following $d$-dimensional (fusion) target density, 
\begin{align}
\label{eq:prod}
\target{}{}{\x} \propto \target{}{1}{\x} \cdots \target{}{C}{\x}, 
\end{align}
where each $\target{}{c}{\x}$ ($c\in\{1,\ldots{},C\}$) is a density (up to a multiplicative constant) representing one of the $C$ distributed inferences we wish to unify. Each $\target{}{c}{\x}$ (which we term a sub-posterior) may in practice itself be represented by a Monte Carlo sample ($\mctarget{}{c}{\x}$), and in this paper we assume we are able to sample (directly and exactly) from each $\target{}{c}{\x}$.

Specific examples of this problem arising naturally in an application include: expert elicitation (\cite{Berger.1980}, \cite{SS:GZ86}), in which the (distributional) views of multiple experts on a topic (or set of parameters) have to be pooled into a single view before a decision-maker can make an informed decision; and, multi-view learning (\cite{multiview.zhao}, \cite{multiview.Li}) and meta-analysis (\cite{Fleiss.1993}, \cite{SIM:SST95}), in which an interpretation could be that we are synthesising multiple inferences on a particular parameter set (computed on datasets which may or may not be of the same type), but the underlying raw data is not directly available for the unified inference. Obtaining the raw data may itself be an insoluble problem due to reasons including the nature of the original publication, data confidentiality, or simply time and storage constraints.

This fusion problem also arises artificially in a number of settings, particularly within modern statistical methodologies tackling `big data'. The computational cost of algorithms such as Metropolis-Hastings, which is an iterative algorithm requiring  full access at every iteration to the full dataset, scale poorly with increasing volumes of data unless a modification is found.

One common modification in light of the challenge of big data is to deploy a `divide-and-conquer' approach (or more accurately termed `fork-and-join' approach (\cite{Stamatakis.2013})). In this setting the full dataset is artificially split into a large number of smaller data sets, inference is then conducted on each smaller data set in isolation, and the resulting inferences are unified (see for instance, \cite{Scott.2016}, \cite{Agarwal.2012}, \cite{Neiswanger.2014}, \cite{Dunson.2013},  \cite{icml:ssld14}, \cite{aistats:sctd16} and \cite{b:lsd17}). The rationale for such an approach is that inference on each small data set can be conducted independently, and in parallel, and so one could exploit large clusters of computing cores to greatly reduce the elapsed time to conduct the full inference. The weakness of these approaches is that the fusion of the separately conducted inferences is inexact. It should be noted that divide-and-conquer methodologies will typically have additional constraints due to hardware concerns -- such as minimising or removing any communication between computing cores to reduce the effect of \textit{latency}. We focus in this paper on the general fusion problem, and so do not fully address the problem in the context of big data (to which we return in subsequent work).

The framework and methodology we outline in this paper (Monte Carlo Fusion) for sampling exactly from (\ref{eq:prod}) can be viewed as a simple rejection sampling scheme on an extended space -- we develop and sample from efficient proposal densities for (\ref{eq:prod}), the samples from which we retain according to an appropriate acceptance probability. The mathematical complication in this paper is in computing the intractable acceptance probability -- which requires the auxiliary simulation of collections of Brownian (or Ornstein-Uhlenbeck) bridges. Our fusion approach provides a principled way to understand the error in existing unification schemes, using a simple linear combination and correction of Monte Carlo samples (analogous to a  traditional meta-analysis approach).

The presentation of this paper broadly follows this pedagogy: In \secref{sec:extended} we present a density on an extended space which admits (\ref{eq:prod}) as a marginal. The remainder of the paper then develops a rejection sampler for the extended density in \secref{sec:extended}. In \secref{sec:brownian_rejection} we develop general theory and methodology for sampling (\ref{eq:prod}) based on a collection of independent Brownian bridges. In \secref{sec:ornstein_rejection} we present a modification of the theory developed in \secref{sec:brownian_rejection} using Ornstein-Uhlenbeck bridges, resulting in sampling efficiencies for the particular (common) setting in which the fusion density is believed to be approximately Gaussian. In \secref{sec:examples} we consider examples of our methodology applied to both light-tailed and heavy-tailed fusion target densities. Finally, in \secref{sec:conclusions} we conclude by discussing the exciting new research directions possible using Monte Carlo Fusion. Much of the technical detail in the paper is suppressed for ease of reading, but can be found in the appendices.


\section{An extended fusion density} \label{sec:extended}
Consider $\target{}{}{\x}$ and $\target{}{c}{\x}$ as described in (\ref{eq:prod}), where $\target{}{}{\x}$ is integrable and we can sample from the density proportional to $\target{}{c}{\x}$.

The following simple observation will form the foundation of our approach: Suppose that $\Trans{}{c}{\x}{\y }$ is the transition density (with respect to Lebesgue measure) of a Markov chain on $\mathbf R^d$ with invariant density proportional to $f_c^2$.

\begin{proposition}
\label{proposition:fusionmeasure}
The $(C+1)d$-dimensional (fusion) density proportional to the integrable function
\begin{align}
g^{}(\x^{(1)}, \ldots, \x^{(C)}, \y )
& = \prod\cores \left[\target{2}{c}{\x^{(c)}} \cdot \Trans{}{c}{\x^{(c)}}{\y}  \cdot  \frac{1}{\target{}{c}{\y}} \right], \label{eq:targ}
\end{align}
admits the marginal density $f$ for ${\boldsymbol y}$.
\end{proposition}

The proof of Proposition \ref{proposition:fusionmeasure} is elementary and is omitted.  The statistical interpretation of Proposition \ref{proposition:fusionmeasure} is simply if it were possible to sample from the $(C+1)d$-dimensional (fusion) density $g^{}$ in (\ref{eq:targ}),  then as a by-product we would obtain a draw from our fusion target density $f$ in (\ref{eq:prod}). How to directly sample from (\ref{eq:targ}) is not clear, even if it were possible to simulate from $p_c(\cdot | \x)$. Our strategy instead will be to use rejection sampling. Two rejection sampling methods (which have differing efficiencies) are provided in Sections \ref{sec:brownian_rejection} and \ref{sec:ornstein_rejection}.



\section{A fusion rejection sampler using Brownian bridges} \label{sec:brownian_rejection}
\subsection{The methodology}

Consider the proposal density for the extended fusion target (\ref{eq:targ}) proportional to the function
\begin{align}
h^{bm}(\x^{(1)}, \ldots, \x^{(C)}, \y)& =
\prod\cores \left[\target{}{c}{\x^{(c)}} \right] \cdot 
 \exp\left(-\frac{C \cdot \| \y - \bar {\x} \|^2}{2T} \right),
\label{eq:proposal}
\end{align}
where $\bar {\x} =C^{-1}\sum_{c=1}^C \x^{(c)}$, and $T$ is an arbitrary positive constant.

Simulation from the proposal $h^{bm}$ can be achieved directly. In particular, $\x^{(1)}, \ldots \x^{(C)}$ are first drawn independently from $f_1, \ldots f_C$ respectively, and then $\y$ is simply a Gaussian random variable centred on ${\bar \x}$. 

In Proposition \ref{proposition:fusionmeasure} we presented a general form of the extended fusion target, in which $\Trans{}{c}{\x}{\y }$ is the transition density (with respect to Lebesgue measure) of a Markov chain on $\mathbf R^d$ with invariant density proportional to $f_c^2$. In this paper we set $\Trans{}{c}{\x}{\y }:= \Trans{dl}{T, c}{\x}{\y}$, the transition density of a double Langevin diffusion for $f_c$ (i.e. the transition density of a Langevin diffusion for $f^2_c$) from $\x$ to $\y$ over a pre-defined (user-specified) time $T>0$. To distinguish the resulting extended fusion target from the general case, we further denote the extended fusion target by $g^{dl}$. In particular, for all $1\le c \le C$, we consider the $d$-dimensional (double) Langevin (DL) diffusion processes $\tX = \{\X{t}{c}, t\in [0,T], c=1,\cdots, C\}$, given by
\begin{align}
\ud\X{t}{c} 
& = \nabla \A{dl}{c}{\X{t}{c}}  \ud t + \ud\W{t}{c}, \label{eq:general_diffusion}
\end{align}
where $\W{t}{c}$ is a $d$-dimensional Brownian motion, $\nabla$ is the gradient operator over $\x$ and
\begin{align}
\A{dl}{c}{\x} & := \log \target{}{c}{\x}. \label{eq:Ax_alphax}
\end{align}
$\X{t}{c}$ has invariant distribution $\target{2}{c}{\x}$, \textit{for any} $t\in [0,T]$ \citep{Hansen.2003}.
We also impose the following standard regularity property (where {\bf div} denotes the divergence operator)

\begin{condition}
\label{condition:2.1}
Define
\begin{align}
 \phifn{dl}{c}{\x} &:= \frac{1}{2} (\|\nabla \A{dl}{c}{\x}\|^2 + \text{\bf div }\nabla\A{dl}{c}{\x}). \label{eq:phi_DL}
\end{align}
There exists constant $\Phifn{bm}{c}>-\infty$ such that for all $\x$ and each $c \in \{1,\cdots, C\}$,
$ \phifn{dl}{c}{\x}  \geq  \Phifn{bm}{c}$.
\label{condition:C2BM}
\end{condition}

Then we have the following proposition which gives a rejection sampling method for $g^{dl}(\x^{(1)}, \ldots, \x^{(C)}, \y )$.

\begin{proposition}
\label{proposition:ratiobound}
Under Condition  \ref{condition:2.1}
we can write
\begin{align}
\label{eq:factrej}
{g^{dl} (\x^{(1)}, \ldots, \x^{(C)}, \y ) \over  h^{bm}(\x^{(1)}, \ldots, \x^{(C)}, \y )} = \left[ \frac{\sqrt{C}}{\sqrt{2\pi T}} \right]^C \times \rho^{bm} \times Q^{bm} \times \prod\cores e^{-T\Phifn{bm}{c}},
\end{align}
where
\begin{align}
\label{eq:Piden}
\rho^{bm} := \rho^{bm}(\x^{(1)}, \cdots, \x^{(C)}) = e^{-\frac{C \sigma^2}{2T}}, \;\;\;\;\; \sigma^2 = C^{-1} \sum\cores \|\x^{(c)} - {\bar {\x}}\|^2,
\end{align}
and
\begin{align}
\label{eq:qiden}
Q^{bm}= {\bf E}_{\overline {\mathbb{W}}}\left(E^{bm}\right),
\end{align}
with ${\overline {\mathbb W}}$ denoting the law of $C$ Brownian bridges $\x^{(1)}_t, \cdots, \x^{(C)}_t$ with $\x^{(c)}_0 = \x^{(c)}$ and $\x^{(c)}_T = \y$ in the time interval $[0,T]$ (noting that these Brownian bridges are independent conditional on the starting and ending points) and
\begin{align}
E^{bm}
 &:= \prod_{c=1}^C \left[
\exp\left\{  -   \int_0^T \left( 
\phifn{dl}{c}{\x^{(c)}_t} -  \Phifn{bm}{c}
\right)
 \ud t 
\right\}
\right].
\end{align}
\end{proposition}
\begin{proof}
From the Dacunha-Castelle representation \citep{DacFlo86} we have that
\begin{align}
\Trans{dl}{T,c}{\x^{(c)}}{\y} 
& = {f_c(\y ) \over f_c (\x^{(c)} )} \times \frac{\sqrt{C}}{\sqrt{2\pi T}} \exp{\left({-
\| \y - \x_c \|^2 \over 2T}
\right)} \times  \; {\bf E}_{\overline {\mathbb{W}}}\left[  \exp\left\{  -   \int_0^T \phifn{dl}{c}{\x^{(c)}_t}  \ud t \right\} \right].
\end{align}
The result then follows from  (\ref{eq:targ}) and (\ref{eq:proposal}) by rearrangement and recalling that $\sum_{c=1}^C \| \y - \x^{(c)} \|^2
= \sum_{c=1}^C \|  \x^{(c)} - {\bar \x} \|^2 + C||\y - {\bar \x}||^2$.
\end{proof}
Here $\rho^{bm}$ and $Q^{bm}$ are both necessarily bounded by $1$, and when interpreted methodologically (see next section) correspond to separate acceptance steps within our rejection sampling framework. An event of probability $\rho^{bm}$ can be simulated by direct computation, and an event of probability $Q^{bm}$ can be simulated using the extensive efficient methodology on Poisson samplers (using an auxiliary diffusion bridge path-space rejection sampler as developed in (for instance) \cite{Beskos.2006,b:bpr06,Beskos.2008}, \cite{MOR:CH13}, \cite{Dai.2014}, \cite{phd:p13}, \cite{Pollock.2016} and \cite{b:pjr15}). Note that there is a trade-off involved in the (user-specified) choice of $T$. For small $T$, $\rho^{bm}$ will likely be small while $Q^{bm}$ is large, whereas for large $T$ the opposite will be true. A small value of $T$ is usually preferred since the computational cost for the diffusion bridge rejection sampling for $Q^{bm}$ is comparatively expensive.

The algorithm for simulating from $f$ (by means of $g$) therefore proceeds as per Algorithm \ref{algorithm:1} (which we term \textit{Monte Carlo Fusion}).

\begin{algorithm}
\SetAlgoLined
{\small
\caption{Monte Carlo Fusion (Brownian Bridge Approach)} \label{algorithm:1}
Initialize a value $T>0$\;
For $c=1, \cdots, C$, simulate $\x_c$ from the density $\target{}{c}{\x}$ and calculate $\bar {\x}$\;
Simulate $\y$ from the Gaussian distribution, with density $\exp\left(-\frac{C \cdot \| \y - {\bar{\x }} \|^2}{2T} \right) $\;
Generate standard uniform random variable $U_1$\;
\eIf{$\log U_1 \leq - \frac{ C\sigma^2 }{2T} $ }{
Generate standard uniform variable $U_2$ and the independent Brownian bridges $\x^{(c)}_{t}, c=1,\cdots, C$ in $[0,T]$, conditional on the starting point $\x_c$ and ending point $\y$\;
\eIf{\begin{eqnarray}\label{eq:Algorithm1U2_al}
U_2 \leq  E^{bm}
\end{eqnarray}}{
Accept and output $\y$ as a sample from $\target{}{}{\x}$\;
\tcp{The event (\ref{eq:Algorithm1U2_al}) can be dealt with via the path-space rejection sampling methods in \cite{Beskos.2005,Beskos.2006,Beskos.2008,Pollock.2016}}
}{Go back to step 2\;}
}{Go back to step 2\;}
}
\end{algorithm}


\subsection{Practical interpretation of the algorithm}

\begin{remark}
Adjustment for the simple average of the sample from sub-densities.
\end{remark}
In the above algorithm, for all $c$, $\x^{(c)}$ is simulated from $\target{}{c}{\x}$ as in other Monte Carlo Fusion algorithms. The proposed combined value $\y$, however, is actually generated from a Gaussian distribution with mean $\bar {\x}$ and covariance matrix $C^{-1} T\mathbf I_d$. Therefore, $\y$ can be viewed as a simple average of the values $\{\x^{(c)},\  c\in C\}$ added to a Gaussian random error term. This algorithm indicates exactly how the simple average of these independent sub-posterior samples can be adjusted as a draw from the target distribution. This adjustment is in the form of the accept/reject step with acceptance probability $\rho^{bm} \times Q^{bm}$. 
\hfill $\square$

\begin{remark}
Implication on Bayesian group decision theory. \label{remark:group}
\end{remark}
In Step 5 of Algorithm \ref{algorithm:1}, we can also write $\log U_1 \leq - \frac{ C\sigma^2 }{2T}$ as
\begin{eqnarray}\label{eq:variancecondition_BM}
\sigma^2 \leq - \frac{2T\log U_1}{C}.
\end{eqnarray}
Note that $\sigma^2$ is actually the {\slshape sample variance} of the simulated starting points $\x^{(c)}$s (strictly speaking, it is the sum of variances for each component of $\x^{(c)}$). Thus in this step, the acceptance condition (\ref{eq:variancecondition_BM}) implies that $\y$ will have a reasonable probability of being accepted as a sample from $\target{}{}{\x}$ only when the variance of $\x^{(c)}$s is small enough.  This coincides with our intuition in group decision making. For example, the small variance of $\{ \x^{(c)},\ c\in C\}$ (satisfying condition (\ref{eq:variancecondition_BM})) means that the decisions/results from each group are similar, so that we can combine these decisions/results in step 7, using (\ref{eq:Algorithm1U2_al}), in Algorithm \ref{algorithm:1}. On the other hand, if the variance of 
$\{ \x^{(c)},\ c\in C\}$
(not satisfying condition (\ref{eq:variancecondition_BM})) is large, then 
$\{ \x^{(c)},\ c\in C\}$ provide contradictory evidence and
should usually be rejected.
 Note that algorithmically, this initial rejection sampling step is efficient since
 early rejection then avoids the need to carry out the more complicated
 (and computationally expensive)
  step 7, via the condition of  (\ref{eq:Algorithm1U2_al}).  \hfill $\square$

\begin{remark}
An extreme case. \label{remark:distance}
\end{remark}
A very interesting extreme case is to choose $T=0$. Then the event (\ref{eq:Algorithm1U2_al}) will certainly happen, but the event $U_1 \leq \exp \left(- \frac{ C\sigma^2 }{2T} \right )$ will occur only if  $\x^{(1)} = \cdots = \x^{(C)} = \y$. Therefore in such an extreme case, Algorithm \ref{algorithm:1} (theoretically) draws a sample $\y$ from
\begin{eqnarray*}
\termprop{bm}{}{\x^{(1)}, \cdots, \x^{(C)},\y} &\propto& \prod\cores \target{}{c}{\y}
\end{eqnarray*}
which is the target distribution. In practice, however, we have to choose $T>0$, since the independent sub-posterior samples $\x^{(c)}$s have zero probability to be the same. 
\hfill $\square$



\section{A fusion rejection sampler using Ornstein-Uhlenbeck bridges} \label{sec:ornstein_rejection} 
\subsection{The methodology}

In the previous section, the proposal density $h^{bm}$  uses a simple average of the sub-posterior samples $\x^{(c)}$ as the mean of the proposal $\y$. Another approach is to consider using a weighted average of the sub-posterior samples $\x^{(c)}$ as the mean of the proposal $\y$. For example, in a typical meta-analysis, a weighted average from different research outputs is typically used as the unification mean, and an individual output with more certainty (or, smaller variance) should consequently have larger weights \citep{Fleiss.1993}.

Denoting $\Bmu$ and $\hatprecon_c$ as the mean and the inverse of covariance matrix estimates for distribution $\target{}{c}{\x}$, respectively. We consider the following more general proposal density
\begin{eqnarray}\label{eq:defhosimple_1}
&&\termprop{ou}{}{\x^{(1)}, \cdots, \x^{(C)}, \y} \propto  \left[\prod\cores \target{}{c}{\x^{(c)}} \right]
\text{\bf etr} \left[ - \frac{1}{2}  \left[ \y - \widetilde{\x} \right]^{\otimes 2} \D \right]
\end{eqnarray}
where
\begin{eqnarray}\label{eq:equationD}
\D &=& \sum\cores \D_c , \;\; \D_c = \Vc^{-1} -  \hatprecon_c \nonumber \\
\Vc  &:=& \Vc(T) = \text{\bf Var } \left(\int_0^T e^{\hatprecon_c(t-T) } d \W{t}{c} \right) =
\frac{\hatprecon_c^{-1}}{2} \left( \mathbf I_d - e^{-2\hatprecon_c T } \right)  
\end{eqnarray}
and
\begin{eqnarray}\label{eq:tildex}
\widetilde {\x}&=&  \D^{-1} \left\{\sum\cores \left( \Vc^{-1} \mc - \hatprecon_c \Bmu\right) \right\} \nonumber\\
\mc &:=& \mc (\x^{(c)},T) = \Bmu + e^{- \hatprecon_c T} (\x^{(c)} - \Bmu).
\end{eqnarray}
It is also straightforward to simulate from $\termprop{ou}{}{\cdot}$, since $\x^{(c)}$s are independently drawn from $\target{}{c}{\x}$ and then $\y$ is generated from a Gaussian distribution with mean $\tilde \x$ and covariance matrix $\D^{-1}$.

Here the vector $\mc $ and the matrix $\Vc$  are actually the mean and covariance matrix, respectively, for the following OU process at time $T$ conditional on the starting point $\x_{0}^{(c)} = \x^{(c)}$ \citep{Masuda.2004},
\begin{eqnarray}\label{eq:Ax_alphax_OU}
&\ud\X{t}{c} = \nabla \A{ou}{c}{\X{t}{c}}  \ud t + \ud\W{t}{c}, \;\; \X{0}{c} \sim \target{2}{c}{\x}& 
\end{eqnarray}
with 
\begin{eqnarray}\label{eq:Ax_OU}
&\A{ou}{c}{\x} = - \displaystyle \frac{ (\Bmu - \x)^{tr} \hatprecon_c (\Bmu - \x)}{2}.   &
\end{eqnarray}

We shall also require the following regularity property.
\medskip
\begin{condition}
\label{condition:3.1}
Define 
\begin{eqnarray}\label{eq:phi_OU_DL}
 \phifn{ou}{c}{\x} 
&= & \frac{1}{2} \left( \|\hatprecon_c (\Bmu - \x)\|^2 - \text{ \bf trace }(\hatprecon_c) \right) .
\end{eqnarray}

For any  $c \in \{1,\cdots, C\}$, there exists $\Phifn{ou}{c}>-\infty$ such that, for all $\x$,
\begin{eqnarray}\label{eq:phi}
 \phifn{dl}{c}{\x} - \phifn{ou}{c}{\x} \geq \Phifn{ou}{c} .
\end{eqnarray}
\label{condition:C2}
\end{condition}

We now have the following result.
\medskip

\begin{proposition}
\label{proposition:OUmainresult}
Define function
\begin{eqnarray}\label{eq:rho}
\rho^{ou} &:=& \rho^{ou}(\x^{(1)}, \cdots, \x^{(C)})  \\
&=& \text{\bf etr}\left\{- \frac{1}{2}  \left[ \bH \D^{-1} +  \sum\cores \bM_{1,c} \left( \mc +  \bM_{1,c} ^{-1} \bM_{2,c}  \Vc \hatprecon_c \Bmu \right)^{\otimes 2}
\right]  \right\}  \nonumber
\end{eqnarray}
with
\begin{eqnarray*}
\bM_{1,c} &=&  e^{2\hatprecon_c T} \hatprecon_c- \Vc^{-1}  \left( \sum\cores \hatprecon_c \right)\D^{-1}  \nonumber \\
\bM_{2,c} &=& \Vc^{-1}\left( \sum\cores \hatprecon_c \right)  \D^{-1} -2\hatprecon_c e^{2\hatprecon_c T} \nonumber \\
\bH &=& \left( \sum\cores  \left(\mc -\Vc \hatprecon_c\Bmu \right)^{\otimes 2}  \Vc^{-1}  \right) \left( \sum\cores  \Vc^{-1} \right) - \left\{ \sum\cores \Vc^{-1} \left(  \mc  - \Vc \hatprecon_c \Bmu \right)  \ \right\}^{\otimes 2} .
\end{eqnarray*}

Under Condition \ref{condition:3.1}
we can write 
\begin{equation}
\label{eq:factrej_ou}
{g^{dl} (\x^{(1)}, \ldots, \x^{(C)}, \y ) \over \termprop{ou}{}{\x^{(1)}, \ldots, \x^{(C)}, \y }} \propto \rho^{ou}(\x^{(1)}, \cdots, \x^{(C)}) \times Q^{ou} \times \prod\cores e^{-T\Phifn{ou}{c}}
\end{equation}
where
\begin{equation}
\label{eq:qiden}
Q^{ou}= {\bf E}_{\overline {\ornst{}{}}}\left(
E^{ou}
\right)
\end{equation}
with $\overline {\ornst{}{}}$ denoting the law of $C$ OU bridges $\x^{(c)}_t, c=1, \cdots, C$ in time interval $[0,T]$. These $\x^{(c)}_t$s are independent conditional on the starting point $\x^{(c)}$ and common ending point $\y$ and
\begin{equation}\label{eq:Eou}
E^{ou}:=  \prod_{c=1}^C \left[ \exp\left\{  -   \int_0^T \left(  \phifn{dl}{c}{\x^{(c)}_t} - \phifn{ou}{c}{\x^{(c)}_t}  -  \Phifn{ou}{c}
\right)
 \ud t 
\right\}
\right]\ .
\end{equation}
\end{proposition}
\begin{proof}
See Appendix.
\end{proof}

Since $\rho^{ou}$ and $Q^{ou}$ in (\ref{eq:factrej_ou}) are always no more than $1$, we have the following rejection sampling algorithm, Algorithm \ref{algorithm:1_OU}.

\begin{algorithm}
\SetAlgoLined
{\small
\caption{Monte Carlo Fusion (Ornstein-Uhlenbeck Approach)} \label{algorithm:1_OU}
Initialise a value $T>0$ and $\Bmu ,\hatprecon_c $\;
For $c=1, \cdots, C$, simulate $\x^{(c)}$ from the density $\target{}{c}{\x}$ and calculate $\widetilde {\x}, \D$\;
Simulate $\y$ from the Gaussian distribution, with mean $\widetilde{\x}$ and covariance matrix $\D^{-1}$ \;
Generate standard uniform random variable $U_1$\;
\eIf{$U_1 \leq \rho^{ou}(\x^{(1)}, \cdots, \x^{(C)})$ (given in the formula (\ref{eq:rho}))}
{
Generate standard uniform random variable $U_2$ and independent OU bridges $\x^{(c)}_t, c=1,\cdots, C$ in $[0,T]$, conditional on the starting point $\x_0^{(c)} = \x^{(c)}$ and ending point $\y$\;
\eIf{\begin{eqnarray}\label{eq:Algorithm1U2_OU}
U_2 \leq \exp \left[ -  \sum\cores \int_0^T \left( \phifn{dl}{c}{\x_{t}^{(c)}} - \phifn{ou}{c}{\x_{t}^{(c)}}  - \Phifn{ou}{c} \right) \ud t \right]
\end{eqnarray}
\tcp{The event (\ref{eq:Algorithm1U2_OU}) can be dealt with via the path-space rejection sampling methods in \cite{Beskos.2005,Beskos.2006,Beskos.2008,Pollock.2016}}
}{
Accept and output $\y$ as a sample from $\target{}{}{\x}$\;
}{go back to step 2}
}
{go back to step 2}
}
\end{algorithm}

In Algorithm \ref{algorithm:1_OU}, $T$ is a tuning parameter as well. If we choose a small value of $T$, the acceptance probability $\rho^{ou} (\x^{(1)}, \cdots, \x^{(C)})$ will be very low. This is noticed by the fact that {\bf trace }$( \bH \D^{-1}) = \infty$ and $\bM_{1,c}$ and $\bM_{2,c}$ are finite matrices, if $T=0$. Therefore in practice, we should choose a reasonably large value $T$. However, if we choose a large $T$, the acceptance probability (\ref{eq:Algorithm1U2_OU}) will be very small. In practice it is usually preferable to choose a small $T$ since the computational cost for the acceptance / rejection step of (\ref{eq:Algorithm1U2_OU}) is much higher.


\subsection{Connection with Consensus Monte Carlo}

In practice, people may employ an approximate version of Algorithm \ref{algorithm:1_OU}, since we can ignore the condition (\ref{eq:Algorithm1U2_OU}) in the rejection step 7, if we choose a small value $T$. In other words, from (\ref{eq:factrej_ou}) of Proposition \ref{proposition:OUmainresult}, we have that for small $T$ the target $g^{dl}(\cdot)$ is approximately equal to a function proportional to
\begin{eqnarray}\label{eq:appro_ou}
\tilde h^{ou} = h^{ou} \cdot \rho^{ou} =  \left[\prod\cores \target{}{c}{\x^{(c)}} \right]
\text{\bf etr} \left[ - \frac{1}{2}  \left[ \y - \widetilde{\x} \right]^{\otimes 2} \D \right] \cdot \rho^{ou}(\x^{(1)}, \cdots, \x^{(C)}) 
\end{eqnarray}
Such an approximation will be very good since $Q^{ou}$, the acceptance probability  (\ref{eq:Algorithm1U2_OU}), will be very close to $1$ with a small $T$. In other words, we only simulate $\y$ from 
$\tilde h^{ou}(\cdot)$ and accept $\y$ as a sample approximately from the target distribution $\target{}{}{\cdot}$.

If we draw $\x_c, c=1, \cdots, C$ independently from each $\target{}{c}{\x}$, respectively. Another naive approach to combine these draws is to use the following linear combination 
\begin{eqnarray}\label{eq:ConsensusAverage}
\y = \left(\sum\cores \hatprecon_c \right)^{-1} \left[ \sum\cores \hatprecon_c \x_c \right].
\end{eqnarray}
In practice, $\hatprecon_c^{-1}$ can be obtained via preliminary analysis. Such a $\y$ can be viewed as a sample approximately from $\target{}{}{\x}$ as well. This is named as Consensus Monte Carlo (CMC) in \citep{Scott.2016}. 

The following lemma tells us how the simulation of $\y$ from $\tilde h^{ou}(\cdot)$ is related to the consensus Monte Carlo sample (\ref{eq:ConsensusAverage}). 

\medskip
\begin{lemma}
If $\target{}{c}{\x}$ is a Gaussian distribution and if we choose $T=\infty$, then simulating $\y$ from $\tilde h^{ou}(\cdot)$ will be the same as the Consensus Monte Carlo method. Both draw samples exactly from the target distribution.
\label{lemma:relationwithScott}
\end{lemma}
\begin{proof}
See Appendix B.
\end{proof}

This lemma tells us why Consensus Monte Carlo does not provide good results. It is because CMC simulates $\y$ from $\tilde h^{ou}(\cdot)$ with $T=\infty$, however, $T$ should be chosen as a small value to achieve better approximation or even use the exact Algorithm \ref{algorithm:1_OU} with the diffusion path rejection sampler. 
  


\section{Simulation studies} \label{sec:examples}
\subsection{Distribution with light tails}
We consider the target distribution $f(x) \propto e^{-x^4/2}$. We choose $C=4$ and $\target{}{c}{x} = e^{-x^4/2C}, c=1,\cdots, C$. It is easy to check that $\phi^{dl}(x) = \frac{1}{2}\left(\frac{4x^6}{C^2} - \frac{6x^2}{C}\right)$ satisfies Condition \ref{condition:C2BM}. On the other hand, for any chosen values  $\Bmu$, $\hatprecon_c$,  $\phifn{ou}{c}{\x}$ in (\ref{eq:phi_OU_DL}) satisfies  Condition \ref{condition:C2} as well. Therefore both Algorithm \ref{algorithm:1} and Algorithm \ref{algorithm:1_OU} can be applied to simulate from $f(x)$. 

We compare the following Monte Carlo methods for the estimation of the density function of  $f(x)$:
\begin{itemize}
\item[1.] Simulating MC samples directly from $f(x)$ via a simple rejection sampling with standard Gaussian distribution as the proposal;
\item[2.] simulating MC samples based on the exact simulation method, Algorithm \ref{algorithm:1} with $T=1$;
\item[3.] simulating MC samples based on the exact simulation method, Algorithm \ref{algorithm:1_OU} with $T=1$;
\item[4.] simulating MC samples based on the Consensus method in \cite{Scott.2016};
\end{itemize}

The density curve estimation results are summarised in Figure
\ref{figure:exp4}. Note that all results are based on $10,000$ realisations.
The black solid cure (Simulation [1.], the true fitted density curve), the blue solid curve (Simulation [2.]) and the pink solid curve (Simulation [3.]) are all exact algorithms and they are almost
identical. The Consensus methods (Simulation [4.] -- the red dashed curve, has very large
biases. Note that both CMC algorithm and Algorithm \ref{algorithm:1_OU} use the same values of $\Bmu$ and $\hatprecon_c$ based on preliminary analysis.

\begin{figure}[h]
\begin{center}
\includegraphics[scale=0.6]{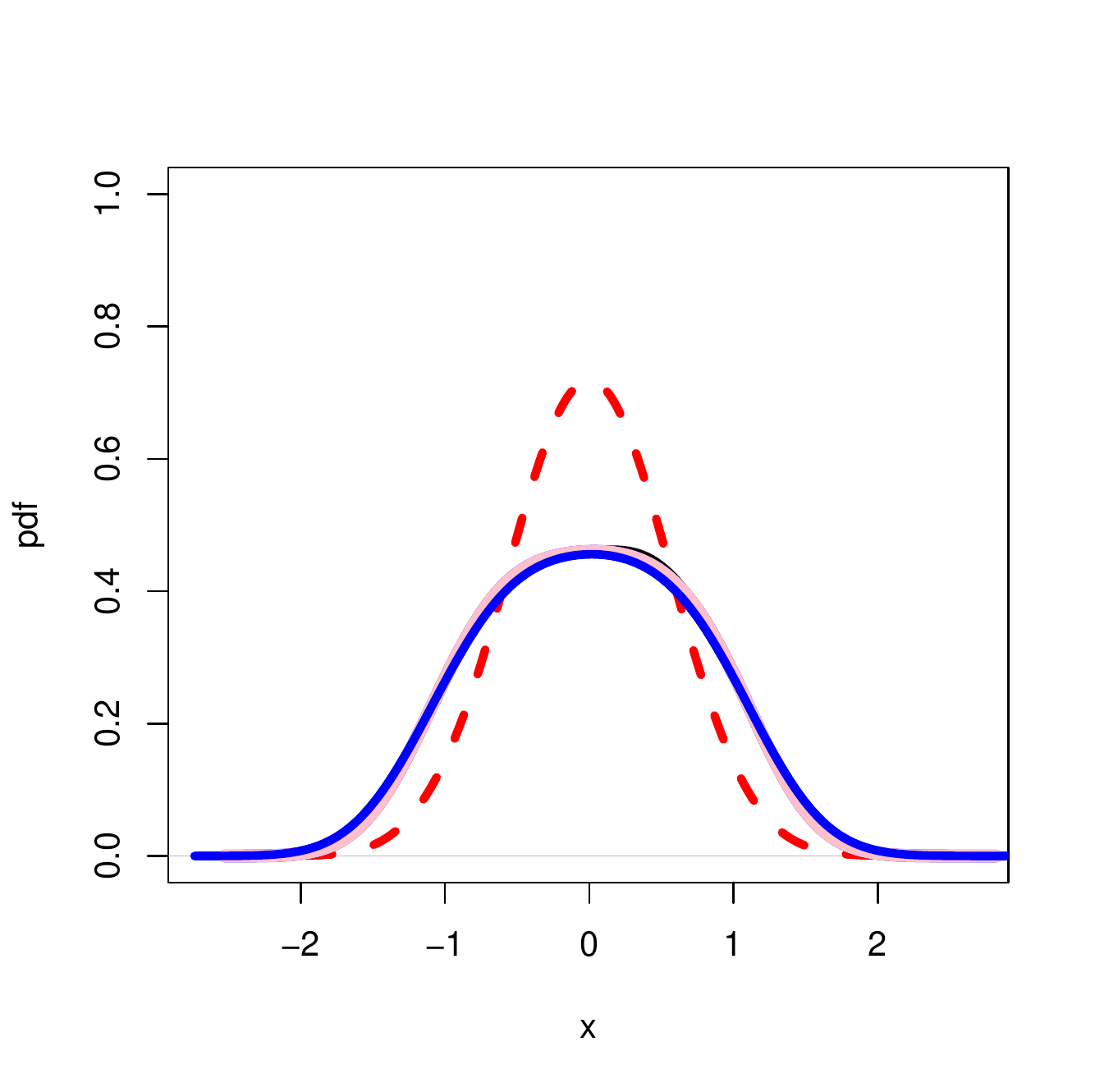}
$\vspace{-0.5cm}$
\caption{Kernel density fitting with bandwidth $0.25$ for density proportional to $e^{-x^4/2}$, based on different Monte Carlo methods.
[1.]-- black solid curve, standard exact MC;
[2.]-- pink solid curve, Algorithm \ref{algorithm:1};
[3.]-- blue solid curve, Algorithm \ref{algorithm:1_OU};
[4.]-- red dashed curve, CMC algorithm.
\label{figure:exp4}}
\end{center}
\end{figure}

The running time of the algorithms are presented in Table \ref{table:runtime}. It seems that Algorithm \ref{algorithm:1_OU} uses the most system running time. This is because Algorithm \ref{algorithm:1_OU} has a smaller acceptance probability (about $0.011$ for the path-space rejection sampling (\ref{eq:Algorithm1U2_OU})) than that in Algorithm \ref{algorithm:1}
(about $0.139$ for the path-space rejection sampling (\ref{eq:Algorithm1U2_al})).

\begin{table}[h]
\begin{center}
\begin{tabular}{cccc}
\hline
Algorithms      & CMC & Algorithm \ref{algorithm:1} & Algorithm \ref{algorithm:1_OU} \\
Running time &   0.05     &         0.36                              &  4.05\\
\hline
\end{tabular}
\caption{System running times in seconds for simulating 10,000 realisations.}\label{table:runtime}
\end{center}
\end{table}

Note that Condition \ref{condition:C2BM} is usually satisfied in most applications, however, Condition \ref{condition:C2} only holds when the target $f(x)$ has lighter tails than Gaussian distributions. We present an example in the following section, where only Condition \ref{condition:C2BM}  is true.


\subsection{Beta distribution}
Consider the target distribution as the Beta distribution with density $\pi(u) \propto u^4 (1-u)$, $u \in [0,1]$, i.e.\ Beta$(5,2)$. To use the proposed algorithms, the support of the target distribution should be in the whole real axis. Therefore we need to use the variable transformation $x = \log(u/(1-u))$ and consider the target distribution as
\begin{eqnarray*}
\target{}{}{x} & \propto & \left[ \frac{\exp(x)}{1+\exp(x)}\right]^5 \left[ \frac{1}{1+\exp(x)}\right]^2 .
\end{eqnarray*}
We decompose $\pi(x)$ into $C=5$ components,
\begin{eqnarray}\label{eq:gkbeta}
\target{}{}{x} &\propto&   \target{}{1}{x} \cdots \target{}{C}{x} \nonumber \\
\target{}{c}{x} &= &\left[ \frac{\exp(x)}{1+\exp(x)}\right] \left[ \frac{1}{1+\exp(x)}\right]^{0.4}.
\end{eqnarray}

Note that for this simple example Condition \ref{condition:C2BM} is satisfied but Condition \ref{condition:C2} is not satisfied. Therefore we compare the following Monte Carlo methods for the estimation of the density function of  Beta$(5,2)$:
\begin{itemize}
\item[1.] Simulating MC samples directly from Beta$(5,2)$ via the simple R command, {\slshape rbeta};
\item[2.] simulating MC samples based on the exact simulation method, Algorithm \ref{algorithm:1} with $T=3$;
\item[3.] simulating MC samples based on the Consensus method in \cite{Scott.2016}, with variable transformation and with the decomposition in (\ref{eq:gkbeta}).
\end{itemize}
For simulations [3.], $\Bmu ,\hatprecon_c $ (also known as the weight of
each consensus sample in CMC algorithm) are chosen, respectively, as the estimated mean and inverse of
variance of $\target{}{c}{\cdot}$, as suggested by \cite{Scott.2016}.

The density curve estimation results are summarized in Figure
\ref{figure:toybetaexample}. Note that all results are based on $10,000$ realisations.
The black solid cure (Simulation
[1.]) and the blue solid curve (Simulation [2.]) are almost
identical, since both of them are based on exact simulation methods. 
Again, Concensus Monte Carlo gives very biased results.

\begin{figure}[h]
\begin{center}
\includegraphics[scale=0.6]{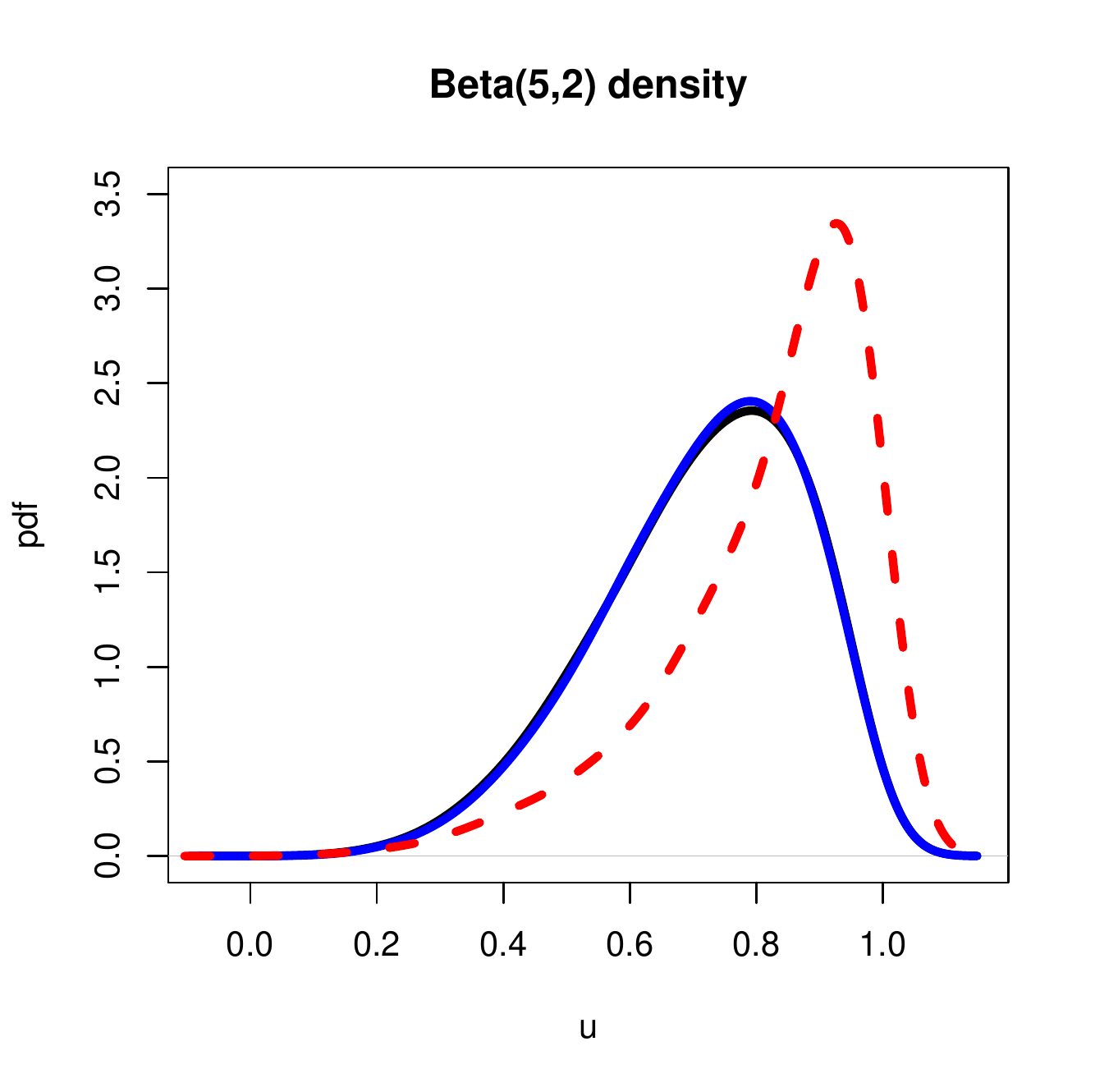}
$\vspace{-0.5cm}$
\caption{Kernel density fitting with bandwidth $0.04$ for Beta$(5,5)$, based on different Monte Carlo methods.
[1.]-- black solid curve, standard exact MC;
[2.]-- blue solid curve, Algorithm \ref{algorithm:1};
[3.]-- red dashed curve, CMC algorithm.
\label{figure:toybetaexample}}
\end{center}
\end{figure}

{\bf Remark 4. Tuning parameters and density decomposition}

We may choose any value $T$ in the proposed algorithms. However, as we mentioned before, value $T$ is a tunning parameter for the efficiency of Algorithm \ref{algorithm:1}, which indeed shown by our simulation results.
The CPU running time for simulating 10,000 realisations based on Algorithm \ref{algorithm:1} for the Beta-example is summarized in Figure \ref{figure:runtime}. The optimum value is about $T=2$. Note that such optimum value can be found in practice via a small amount of preliminary simulations.

\begin{figure}[h]
\begin{center}
\includegraphics[scale=0.5]{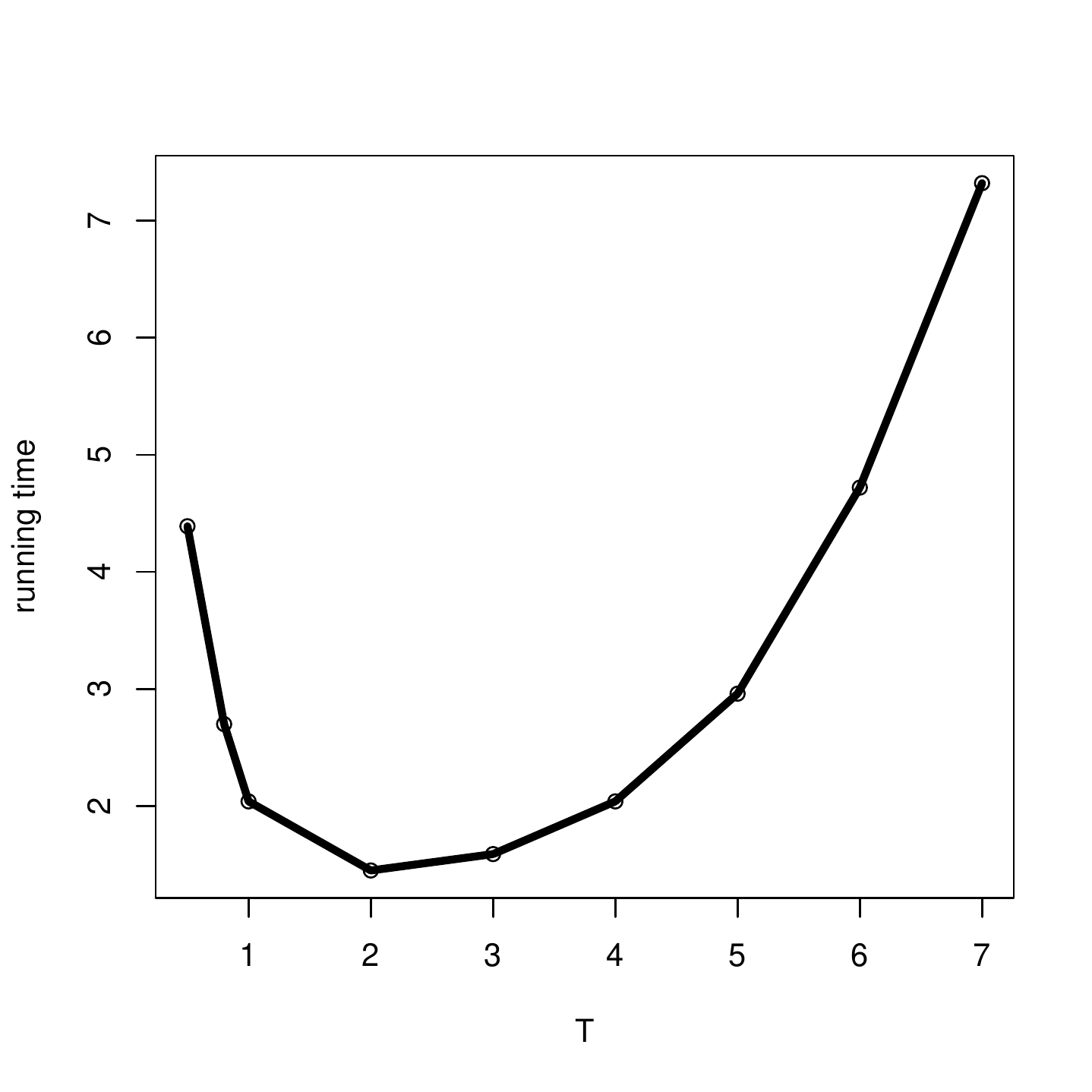}
$\vspace{-0.5cm}$
\caption{CPU time (in seconds) for  Algorithm \ref{algorithm:1}, based on different $T$.
\label{figure:runtime}}
\end{center}
\end{figure}

In addition, when we split the target $f$ into $f\propto f_1\cdots f_C$, we actually chose $f_1=\cdots =f_C = f^{1/C}$, since such decomposition will always give the smallest variation for $\x^{(c)}, c=1,\cdots,C$, as suggested in \cite{Dai.2017}.


\FloatBarrier
\section{Conclusion} \label{sec:conclusions}

In this paper we have introduced a novel theoretical framework, and direct methodological implementation (Monte Carlo Fusion), to address the common (but challenging) `fusion' problem -- unifying distributed analyses and inferences on shared parameters from multiple sources, into a single coherent inference -- by viewing it as a simple rejection sampler on an extended space (\secref{sec:extended}), and developing an appropriate sampling mechanism (\secref{sec:brownian_rejection}).

Our fusion approach is not only the first paper to answer in a principled manner how to combine samples from multiple sources, but (as shown in \secref{sec:ornstein_rejection}), also provides a principled approach to understand the errors that arise in other existing unification schemes (such as those in big data divide-and-conquer approaches). The errors in existing unification schemes can be considerable, even for simple one dimensional unification targets (such as those we consider in \secref{sec:examples}).

Characterising the error in existing unification schemes is possible by setting the (proposal) sampling mechanisms of those schemes within our framework, and finding a representation for the remaining error (which we could remove by further acceptance or rejection). This opens interesting avenues of research in which existing unification schemes are adapted within our framework into efficient proposal mechanisms for our extended fusion target density (\ref{eq:targ}).

A number of avenues to apply our work directly to interesting applications are possible. In addition to those discussed in \secref{sec:intro} (namely expert elicitation, multi-view learning and meta-analysis), other application areas include Bayesian group decision theory (see Remark \ref{remark:group}) and Bayesian sensitivity analysis. In the case of Bayesian sensitivity analysis we need to assess a large number of prior distributions, but existing methods address this by using approximations \citep{Tan.2015}.

A key avenue for (on-going) future research is to fully explore how to use the Monte Carlo Fusion framework we introduce to modify, and remove approximation, from existing Monte Carlo methods (particularly within the big data setting) that use the divide-and-conquer (fork-and-join) strategies described in \secref{sec:intro}. In the particular setting of big data the unification of distributed inferences is only part of the problem -- additional constraints are imposed due to practical computational and hardware concerns (such as avoiding as far as possible communication between computing cores). As such, key future research will focus on how to implement Monte Carlo Fusion with this particular set of constraints (as direct implementation is not possible).

Another interesting big data direction would be to blend the multi-core approach of divide-and-conquer strategies, with state-of-the-art single-core approaches for big data (for instance, \cite{Pollock.2016}). 

\section*{Acknowledgements}

The authors would like to thank the Isaac Newton Institute for Mathematical Sciences for support and hospitality during the programme ``Scalable inference; statistical, algorithmic, computational aspects (SIN)'' when work on this paper was undertaken. MP and GOR were supported by the EPSRC [grant number EP/K014463/1]. GOR was additionally supported by the EPSRC [grant numbers EP/K034154/1, EP/D002060/1].


\appendix


\section{Proof of Proposition \ref{proposition:OUmainresult}}

To prove the proposition, we first introduce a lemma about a density proportional to the following function
\begin{eqnarray}\label{eq:proposalou}
&& \tilde g^{ou}( \x^{(1)}, \ldots, \x^{(C)},\y ) \nonumber \\
&=& \prod\cores \left[ \target{}{c}{\x^{(c)}} \Trans{ou}{T,c}{\x^{(c)}}{\y} \exp\left( \A{ou}{c}{\x^{(c)}} - \A{ou}{c}{\y}  \right) \right]
\end{eqnarray}
where $ \Trans{ou}{T,c}{\x^{(c)}}{\y} $ is the transition density from $\x^{(c)}$ at time $0$ to $\y$ at time $T$, for the OU process given by (\ref{eq:Ax_alphax_OU}).

\medskip

\begin{lemma}
{\bf The expression of $\tilde g^{ou}$.} The formula in (\ref{eq:proposalou}) can be rewritten as
\begin{eqnarray}\label{eq:defhosimple_1}
&&\tilde g^{ou}\left({\x^{(1)}, \cdots, \x^{(C)}, \y}\right) \nonumber \\
&\propto & \left[\prod\cores \target{}{c}{\x^{(c)}} \right]
\text{\bf etr} \left[ - \frac{1}{2}  \left[ \yT - \widetilde{\x} \right]^{\otimes 2} \D \right]
\rho^{ou} (\x^{(1)}, \cdots, \x^{(C)})
\end{eqnarray}
where $\widetilde{\x}$ and $\rho^{ou}$ are given by (\ref{eq:tildex}) and (\ref{eq:rho}), respectively.
\hfill $\square$
\label{lemma:OUresult}
\end{lemma}

\begin{proof}
See the supplementary file.
\end{proof}

Now we prove Proposition \ref{proposition:OUmainresult}.
\begin{proof}
If we denote $\overline{\lange{}{}}$ as the law of $C$ $d$-dimensional Langevin diffusion bridges given in (\ref{eq:general_diffusion}), with starting points $\x^{(c)}$ and common ending point $\y$, the result of Proposition \ref{proposition:ratiobound} can be written as
\begin{equation}
\label{eq:factrej}
{g^{dl} (\x^{(1)}, \ldots, \x^{(C)}, \y ) \over  h^{bm}(\x^{(1)}, \ldots, \x^{(C)}, \y )} \times \ud \overline{\lange{}{}}(\vec{\x}) \propto \rho^{bm} \times E^{bm} \times \prod\cores e^{-T\Phifn{}{c}} \times \ud \overline{\mathbb W}(\vec{\x})
\end{equation}
where $\vec{\x} = \{ \x^{(c)}_t, c=1,\cdots, C, t\in[0,T] \}$ are typical diffusion bridge paths, with starting points $\x^{(c)}$ and common ending point $\y$.

If we consider a very special case where each $\target{}{c}{\x}$ is a Gaussian density $\exp(\A{ou}{c}{\x})$, we immediately have that the target $g^{dl}$ becomes
$$
g^{ou} = \prod\cores\left[ e^{2 \A{ou}{c}{\x^{(c)}}} \Trans{ou}{T,c}{\x^{(c)}}{\y} \frac{1}{e^{\A{ou}{c}{\y}} } \right]
$$

In addition, the proposal density $h^{bm}$ in Propostion \ref{proposition:ratiobound} becomes 
$$
 {\hbar}^{bm}(\x^{(1)}, \ldots, \x^{(C)}, \y )  = \prod\cores \left[ e^{\A{ou}{c}{\x^{(c)}}} \right]  e^{-\frac{\| \y - \bar \x\|^2}{2T}},
$$
and further the rejection sampling ratio in Proposition \ref{proposition:ratiobound} becomes
\begin{equation}
\label{eq:factrej_ou2}
{g^{ou} (\x^{(1)}, \ldots, \x^{(C)}, \y ) \over  \hbar^{bm}(\x^{(1)}, \ldots, \x^{(C)}, \y )}  \times \ud \overline{\ornst{}{}}(\vec{\x})  \propto \rho^{bm} \times E^{ou}_* \times \ud \overline{\mathbb W}(\vec{\x}) 
\end{equation}
\begin{equation}
E^{ou}_*:=
 \prod_{c=1}^C \left[
\exp\left\{  -   \int_0^T \phifn{ou}{c}{\x^{(c)}_t} 
 \ud t 
\right \}
\right]\ .
\end{equation}

Now we have,
\begin{eqnarray*}
{g^{dl} \over  h^{ou}} \cdot \frac{\ud \overline{\lange{}{}}}{\ud \overline{\ornst{}{}}}&=& {g^{dl} \over  h^{bm} }
\frac{\ud \overline{\lange{}{}}}{\ud \overline{\mathbb W}}  \cdot { \hbar^{bm} \over  g^{ou}} 
\frac{\ud \overline{\mathbb W}}{\ud \overline{\ornst{}{}}} \cdot {g^{ou} \over h^{ou}}\cdot { h^{bm} \over  \hbar^{bm}} \\
&\propto& \frac{E^{bm}}{E^{ou}_*} \cdot \frac{\prod\cores\left[ e^{2\A{ou}{c}{\x^{(c)}}} \Trans{ou}{T,c}{\x}{\y} \frac{1}{e^{\A{ou}{c}{\y}} } \right] }
{ \left[\prod\cores \target{}{c}{\x^{(c)}} \right]
\text{\bf etr} \left[ - \frac{1}{2}  \left[ \y - \widetilde{\x} \right]^{\otimes 2} \D \right]} \cdot \frac{\left[\prod\cores \target{}{c}{\x^{(c)}} \right]}{\left[\prod\cores e^{\A{ou}{c}{\x^{(c)}}}\right]} \\
&=& \frac{E^{bm}}{E^{ou}_*} \cdot \frac{\tilde g^{ou}(\x^{(1)}, \cdots, \x^{(C)}, \y)} { \left[\prod\cores \target{}{c}{\x^{(c)}} \right]
\text{\bf etr} \left[ - \frac{1}{2}  \left[ \y - \widetilde{\x} \right]^{\otimes 2} \D \right]}.
 \end{eqnarray*}

 Lemma \ref{lemma:OUresult} and expressions of $E^{bm}$ and $E_*^{ou}$ immediately gives
  \begin{eqnarray*}
{g^{dl} \over  h^{ou}}  &=& \frac{E^{bm}}{E^{ou}_*} \cdot \rho^{ou}  \propto  E^{ou} \cdot \rho^{ou}
    \end{eqnarray*}
where $E^{ou}$ given in (\ref{eq:Eou}).
 \end{proof}
 
  
\section{ Proof of Lemma \ref{lemma:relationwithScott}}

We consider the formula of $\rho^{ou}$ in (\ref{eq:rho}). If $T=\infty$, we have
\begin{eqnarray*}
\mc = \Bmu, \quad\quad \quad \Vc = \frac{\hatprecon_c^{-1}}{2}
\end{eqnarray*}
and 
\begin{eqnarray*}
\D &=& \sum\cores \hatprecon_c, \nonumber \\
\bM_{1,c} &=& e^{2\hatprecon_c T} \hatprecon_c - 2\hatprecon_c\\
\bM_{2,c} &=& 4\hatprecon_c - 2 e^{2\hatprecon_c T} \hatprecon_c = -2 \bM_{1,c} \\
\end{eqnarray*}
Therefore
\begin{eqnarray*}
 && \lim_{T\rightarrow \infty}\bM_{1,c} \left( \mc + \bM_{1,c}^{-1} \bM_{2,c} \Vc \hatprecon_c \Bmu \right)^{\otimes 2} \nonumber \\
&=&   \lim_{T\rightarrow \infty} \left( \bM_{1,c}^{1/2} \Bmu + \bM_{1,c}^{-1/2} \bM_{2,c} \frac{\Bmu}{2} \right)^{\otimes 2} \nonumber \\
&=& \mathbf 0
\end{eqnarray*}
and further $\rho^{ou}(\x^{(1)}, \cdots, \x^{(C)})$ becomes a value not depending on $\X{0}{1:C}$ at all, as $T\rightarrow \infty$. Therefore with $T=\infty$ the density function
$\tilde h^{ou}(\cdot)$ becomes
\begin{eqnarray*}
\tilde h^{ou}{\cdot}(\x^{(1)}, \cdots, \x^{(C)}, \y) &\propto & \left[\prod\cores \target{}{c}{\x^{(c)}} \right]
\text{\bf etr} \left[ - \frac{1}{2}  \left[ \y - \widetilde{\x} \right]^{\otimes 2} \left(\sum\cores \hatprecon_c \right) \right]
\end{eqnarray*}
with
\begin{eqnarray*}
\widetilde{\x} = \left(\sum\cores \hatprecon_c \right)^{-1} \left\{ \sum\cores \hatprecon_c \Bmu \right\}
\end{eqnarray*}
Then we can generate $\y$ as, with some standard Gaussian random errors $\boldsymbol \epsilon$,
\begin{eqnarray*}
\y &=&   \widetilde{\x}  + \left(\sum\cores \hatprecon_c \right)^{-1/2} \cdot \boldsymbol \epsilon\\
&=& \left(\sum\cores \hatprecon_c \right)^{-1} \left\{ \sum\cores \hatprecon_c \Bmu + \left(\sum\cores \hatprecon_c\right)^{1/2}  \boldsymbol \epsilon \right\} \\
&\stackrel{distr.}{=}{} & \left(\sum\cores \hatprecon_c \right)^{-1} \left\{ \sum\cores \hatprecon_c \Bmu + \sum\cores \hatprecon_c^{1/2} \boldsymbol \epsilon_c \right\}\\
&=& \left(\sum\cores \hatprecon_c \right)^{-1} \left\{ \sum\cores \hatprecon_c \left( \Bmu + \hatprecon_c^{-1/2} \boldsymbol \epsilon_c \right) \right\}
\end{eqnarray*}
where $\stackrel{distr.}{=}{} $ means `equal in distribution' and $\boldsymbol \epsilon_c, c=1,\cdots, C$ means $C$ independent standard normal vectors.

By noticing that $\Bmu + \hatprecon_c^{-1/2} \boldsymbol \epsilon_c$ has the same distribution as $\x^{(c)}$ if $\target{}{c}{\x}$ is a Gaussian distribution, the lemma is proved.



\end{document}